\documentclass[conference]{IEEEtran} 
\IEEEoverridecommandlockouts
\def\BibTeX{{\rm B\kern-.05em{\sc i\kern-.025em b}\kern-.08em
    T\kern-.1667em\lower.7ex\hbox{E}\kern-.125emX}}
    
\usepackage{verbatim}
\usepackage{multirow}
\usepackage[dvips]{graphics} 
\usepackage{graphicx} 
\usepackage{subfigure}
\usepackage{times} 
\usepackage[font=small,labelfont=bf]{caption} 
\usepackage{epsfig}
\usepackage{float}
\usepackage[utf8]{inputenc}
\usepackage[english]{babel} 
\usepackage[noadjust]{cite}
\usepackage{enumerate}
\usepackage{amssymb}
\usepackage{amsthm}
\usepackage{multirow}
\usepackage{latexsym}
\usepackage{color}
\usepackage{float}
\usepackage{amsmath}      
\usepackage{cite}
\usepackage[noend]{algpseudocode}
\usepackage{epstopdf}
\usepackage{tikz}
\usepackage{algorithm}
\usepackage{mathtools}

\allowdisplaybreaks
\usepackage{etoolbox}
\newtheorem{theorem}{Theorem}
\newtheorem{remark}{Remark}
\makeatletter
\patchcmd{\@begintheorem}{\textit}{\textbf}{}{}

\usepackage{xcolor}
\usepackage{color}

\usepackage{mathtools, cuted}
\usepackage{lipsum, color}
\usepackage{amssymb,amsmath,amsthm,enumitem}

\newcommand*{\algrule}[1][\algorithmicindent]{%
  \makebox[#1][l]{%
    \hspace*{.2em}
    \vrule height .75\baselineskip depth .25\baselineskip
  }
}

\newcount\ALG@printindent@tempcnta
\def\ALG@printindent{%
    \ifnum \theALG@nested>0
    \ifx\ALG@text\ALG@x@notext
    \else
    \unskip
    \ALG@printindent@tempcnta=1
    \loop
    \algrule[\csname ALG@ind@\the\ALG@printindent@tempcnta\endcsname]%
    \advance \ALG@printindent@tempcnta 1
    \ifnum \ALG@printindent@tempcnta<\numexpr\theALG@nested+1\relax
    \repeat
    \fi
    \fi
}
\patchcmd{\ALG@doentity}{\noindent\hskip\ALG@tlm}{\ALG@printindent}{}{\errmessage{failed to patch}}
\patchcmd{\ALG@doentity}{\item[]\nointerlineskip}{}{}{} 
\makeatother 
\usepackage[T1]{fontenc}

\usepackage[font=small,skip=0pt]{caption}

\usepackage{lineno}

\begin{document}

\title{Analysis and Optimization of the Latency Budget in Wireless Systems with Mobile Edge Computing} 
\author{\IEEEauthorblockN{  Suraj Suman\IEEEauthorrefmark{1},   \v Cedomir Stefanovi\' c\IEEEauthorrefmark{1},    Strahinja Do\v sen\IEEEauthorrefmark{2}, and Petar Popovski\IEEEauthorrefmark{1}\\
\IEEEauthorblockA{
\IEEEauthorrefmark{1}Department of Electronic Systems, Aalborg University, Denmark}} 
\IEEEauthorblockA{
\IEEEauthorrefmark{2}Department of Health Science and Technology, Aalborg University, Denmark \\
Email: \{ssu,   cs\}@es.aau.dk, sdosen@hst.aau.dk, petarp@es.aau.dk
}
}
\maketitle

\begin{abstract}
We present a framework to analyse the latency budget in wireless systems with Mobile Edge Computing (MEC). 
Our focus is {on teleoperation and telerobotics, as use cases that are representative of mission-critical uplink-intensive IoT systems with requirements on low latency and high reliability. The study is motivated by a general question: \emph{What is the optimal compression strategy in reliability and latency constrained systems?}}
{We address this question by studying the latency of an uplink connection from a multi-sensor} IoT device to the base station. {This is a} critical link tasked with a timely and reliable transfer of potentially significant amount of data from the multitude of sensors.
We {introduce a comprehensive model for the latency budget, } incorporating data compression and data transmission. {The uplink} latency is a \emph{random variable} whose distribution depends on the computational capabilities of the device and on the properties of the wireless link.
 {We formulate two optimization problems corresponding to two transmission strategies: (1) \emph{Outage-constrained,} and (2) \emph{Latency-constrained}. We derive the optimal system parameters under a reliability criterion.
We show that the obtained results are superior compared to the ones based on the optimization of the expected latency. }


\end{abstract}


\section{Introduction}
\label{sec:introduction}


The recent advancements in wireless networking and computing systems pave the way for novel mission-critical Internet-of-Things (IoT) use-cases that rely on  {Mobile Edge Computing (MEC)}. {Representative use cases include 
teleoperation and telerobotics, both  characterized by an uplink-intensive communication from a multi-sensory IoT device, comprising audio, video and haptic data traffic~\cite{Fettweis_TI}. This traffic is often mission-critical, such that the key requirements include
low latency and high reliability.}
However, a wearable IoT device that streams the {multi-sensory data in the uplink is resource-constrained, not least with respect to the available computational power. Furthermore, reliability and latency are challenged by the variable and error-prone wireless link.}

In this paper, we examine a scenario in which a device located in the remote environment interacts with the {Base Station (BS), equipped with a MEC server.}
We focus on the segment of the uplink connection between the device and the BS, and analyse the uplink latency that comprises the time elapsed in data compression followed by the transmission.
The key contributions of this work are: 
\begin{enumerate}
\item {We derive a} tractable model of the uplink latency {as a random variable (RV),} relating the lossless-compression ratio and link-outage probability.
{Specifically, we obtain the probability distribution of the latency.}  
\item {We consider two different transmission strategies: (1) \emph{Outage-constrained transmission,} and (2) \emph{Latency-constrained transmission} and formulate the respective optimization problems. These problems are shown to be non-convex, and are transformed into convex ones. This allows to find the optimal latency and optimal outage, respectively.}
\item The obtained results are compared with the ones of the analogous optimization problem that {are tailored to the expected values for latency, as done in the prior works. The results show that the proposed approach is superior in terms of reliability and latency.}
\item {Interestingly, } the results  reveal that the data compression is not always beneficial {and its utility} depends on the computational capability of the device and the reliability requirements. 
\end{enumerate}

The rest of the paper is organised as follows. 
This section is concluded by a brief overview of the related work.
Section~\ref{sec:syst_model} introduces the system model.
Section~\ref{sec:analysis} models  the uplink latency and derives its probability distribution function (PDF).
The optimization problems relevant for the system design are discussed in Section~\ref{sec:optimization}.
The evaluation is presented in Section~\ref{sec:evaluation}, followed by concluding remarks in Section~\ref{sec:conclusions}.

\subsection*{Related Work}

Time delay significantly affects the performance of teleoperation systems~\cite{effect_of_delay}. 
The use of fiber-wireless (FiWi) networks have been envisioned in~\cite{onu_1,onu_3} to reduce latency in teleoperation applications,  where the wireless front ends, such as base stations and WiFi access points, are integrated with optical network units.
These works consider the average end-to-end delay while analysing the latency-constrained teleoperation scenario.  
In~\cite{NFV_JSAC}, 5G network architecture based on the network function virtualization (NFV) technology is presented to support the implementation of Tactile Internet (TI) applications.
Using NFV-based TI architecture, a utility function-based model is reported in~\cite{IoT_TI} to evaluate the performance of the NFV-based TI by considering the resolution of the human perception and the network cost of completing services.
The utility function depends on the average round-trip delay, network link bandwidth, and node virtual resource consumption.

Recently, mobile edge computing (MEC) was used in a TI application~\cite{MEC_TII, hybrid_caching_TI, TelSurg}, where the tasks are offloaded to a MEC server for processing. 
An MEC-based TI system was designed in~\cite{MEC_TII} to satisfy the quality-of-experience. 
A hybrid edge-caching scheme for heterogeneous TI network is presented in \cite{hybrid_caching_TI}, where the average end-to-end latency, comprising transmission times between users to edge nodes and edge nodes to central cloud servers, is optimized.
A real time  architecture for telesurgery application is presented in~\cite{TelSurg}, where real-time telesurgery is envisioned by employing cloud and MEC networks using software-defined networking as infrastructure to satisfy the average end-to-end latency.  


The works reported in \cite{onu_1,onu_3,NFV_JSAC,IoT_TI,MEC_TII, hybrid_caching_TI, TelSurg} treat the latency through its expected value, which poses limitations on the applicability of these approaches.
Specifically, as shown in the present paper, the latency is a RV, which also impacts the communication reliability.
Consequentially, system design {that relies upon the optimization of} the expected latency does not ensure high-reliability, as it leads to underprovisioning. 
Furthermore, the reliability of data transmission with tolerable packet loss has also not been considered.
We also note that, to the best of our knowledge,  {in these works} the data-compression latency {has neither been} explored nor modelled.

\section{ System Model }
\label{sec:syst_model}
The system model is shown in Fig.~\ref{fig:sys_mod_MEC_TI}.
It consists of  {a multi-sensory IoT device, such as a robot,} located in a remote environment and a BS  equipped with MEC server, communicating over a wireless channel.  
The BS collects the data from the device and process it on the MEC server. 
The volume of the data may be significant, potentially requiring compression before the data is transmitted. 
Thus, the total  {latency budget in the uplink consists of parts pertaining to both data compression and transmission.}
The latency of data compression depends upon the compression ratio and the processing capability of the device. 
Regarding the data transmission, the randomness in wireless channel restricts the data rate that can be transmitted reliably.
In the next section, we perform analysis of the total uplink latency, which is the critical contributor to the total latency in the system.


\begin{figure}[!t] 
\centering 
{{\includegraphics[width=0.75\columnwidth]{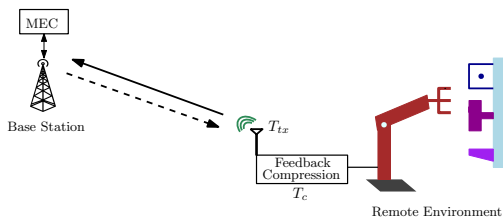}}} 
\vspace{0.25cm} 
\caption{System model; the focus  of the paper is on the uplink connection ($T_c$: data compression latency, $T_{tx}$: data transmission latency).} 
\label{fig:sys_mod_MEC_TI} 
\end{figure}  

\section{ Latency Analysis}
\label{sec:analysis}

\subsection{Latency in Data Compression } 

The latency of data compression depends on the data volume and computational properties of the device's processor. 
Specifically, the time elapsed $T_\text{c}$ in compressing volume of data $D_0$ is given as \cite{jsac_comp_model} 
\begin{equation} \label{eq:comp_time}
T_\text{c} = \frac{D_0 X}{f_\text{R}}
\end{equation}
where $X$ is the number of CPU cycles required to compress one bit of data, and $f_\text{R}$ is the frequency (i.e., clock speed) of the processor.
A recent study shows that $X$ is stochastic in nature~\cite{MEC_random_var}, i.e., it is a RV that follows the Gamma distribution $X \sim \mbox{Gamma}(\kappa, \beta)$~\cite{MEC_gamma_1, MEC_gamma_2}.
Specifically, 
\begin{equation} \label{eq:no_of_cycles}
f_X(x) = \frac{1}{ \beta^\kappa \Gamma(\kappa)} x^{\kappa-1} \exp(-{x}/{\beta})
\end{equation}
where $\kappa$ and $\beta$ are respectively the shape and scale parameters, and $\Gamma(s)= \int_{0}^{\infty} t^{s-1} e^{-t} dt$ is the Gamma function.
Note that $\mathbb{E}[X] = \kappa  \beta$.

Thus, $T_\text{c}$ is also a RV, whose PDF is derived as
\begin{equation} \label{eq:t_mec} 
\begin{split}
    f_{T_\text{c}}(t)  &   = \left( \frac{f_{R}}{D_0 \beta}\right)^\kappa \frac{1}{\Gamma(\kappa)} t^{\kappa-1} \exp \left(-\frac{t f_\text{R}}{D_0  \beta}\right).   
\end{split}
\end{equation}
We assume that lossless compression is performed, so that the original, raw data can be reconstructed perfectly.\footnote{Techniques like Huffman, run-length, and Lempel-Ziv encoding efficiently and losslessly compress the raw data~\cite{compression_survey}.}
For lossless compression, the average number of CPU cycles required  to compress one bit of raw data is given as~\cite{comm_lett_comp,jsac_comp_model} 
\begin{equation}  \label{eq:compresiion_ratio}
\mathbb{E}[X] = \kappa \beta = \exp(Q \psi) - \exp(\psi) = C(Q)
\end{equation} 
where $Q\geq 1$ is the compression ratio (i.e., the ratio of the sizes of raw and compressed data) and $\psi$ is a positive constant. 
Using \eqref{eq:compresiion_ratio}, the PDF of compression time ${T_\text{c}}$ becomes 
\begin{equation} \label{eq:t_mec} 
\begin{split}
    f_{T_\text{c}}(t; Q)  &   =   \left( \frac{f_\text{R}}{D_0 \frac{C(Q)}{\kappa}}\right)^{\kappa} \frac{1}{\Gamma(\kappa)} t^{\kappa-1} \exp \left(-\frac{t f_\text{R}}{D_0  \frac{C(Q)}{\kappa}}\right)   
\end{split}
\end{equation}
The cumulative distribution function (CDF) of ${T}_\text{c}$ for compression ratio $Q$ is given as
\begin{equation}
F_{T_\text{c}}(t; Q) = \frac{\Gamma \left(\kappa,  \frac{t}{\frac{ D_0 C(Q)}{\kappa f_R }} \right)}{\Gamma(\kappa)}. 
\end{equation}
The expected value of time elapsed in compression for compression ratio $Q$ is given as 
\begin{equation} 
\bar{T}_\text{c}(Q) = \mathbb{E}[T_\text{c}] = \frac{D_0 \mathbb{E}[X]}{f_\text{R}}  =  \frac{D_0 C(Q)}{f_\text{R}}.
\end{equation}

\subsection{ Latency in Data Transmission }

The wireless channel is assumed to feature a quasi-static fading, where the channel gain $h$ is a Rayleigh RV independently and identically distributed over the time-slots.
Hence, the channel power $g=|h|^2$ follows the exponential distribution, and we assume that $g \sim \exp(1)$.   
The device transmits with power $P_\text{tx}$, so the signal-to-noise ratio ($\mbox{SNR}$) $\gamma$ at the receiver located at distance $d$ away is given as 
\begin{align}   \label{eq:snr}
\gamma(P_\text{tx},d, B) & = \frac{\mathcal{K}_0 P_\text{tx} |h|^2}{  d^2 N_0 B} =  \gamma_0(P_\text{tx},d, B) \cdot g 
\end{align}
where $\mathcal{K}_0$ is a constant accounting the Friss equation parameter, $N_0$ is the power spectral density of noise,  $B$ is the allocated bandwidth, and  $\gamma_0(P_\text{tx}, d, B) =  \frac{P_\text{tx}}{\mathcal{K}_0 d^2 N_0 B}$.    

Further, we assume that the transmitter wants to {send} at the rate that guarantees $\epsilon$-outage at the receiver~\cite{goldsmith2005wireless}; this rate {depends on the statistics of the }SNR.
The outage probability $\epsilon$ characterizes the probability of data loss in case of deep fading, when the transmission cannot be decoded.
 We assume that the data is correctly received if the instantaneous received SNR is not lower than $\gamma_\text{th}$;  {otherwise, an outage is declared.}
For a threshold SNR $\gamma_\text{th}$ with outage $\epsilon$, the rate $R(\epsilon)$ is 
\begin{equation}  \label{eq:rate}
R(\epsilon) = \log_2(1+\gamma_\text{th})
\end{equation}
where  outage probability $\epsilon$ is given as
 \begin{equation}\label{eq:epsilon}
 \epsilon = \Pr \left[ \gamma < \gamma_\text{th} \right] = \Pr \left[ g < \frac{\gamma_\text{th}}{\gamma_{0}} \right] = 1 - \exp\left(-\frac{\gamma_\text{th}}{\gamma_{0}}\right).
\end{equation}
From \eqref{eq:rate} and \eqref{eq:epsilon},  $R(\epsilon)$ can be rewritten as 
\begin{equation} \label{eq:rate_epsilon} 
R(\epsilon) = \log_2(1+\gamma_\text{th}) = \log_2 \left( 1+\gamma_{0} \ln \left(\frac{1}{1-\epsilon} \right) \right).
\end{equation}

The number of channel uses $N_t$ required  to transmit data volume $D$ with outage $\epsilon$ is given as
\begin{equation}  \label{eq:channel_use} 
N_\text{tx}(\epsilon) = \frac{D}{R(\epsilon)} 
\end{equation}
and the time elapsed in transmitting is 
\begin{equation} \label{eq:tx_time} 
T_\text{tx}(\epsilon) = T_0 N_\text{tx}(\epsilon)   
\end{equation}
where $T_0$ denotes the duration of a channel use and is determined by the available bandwidth that is fixed. 

\begin{remark} \label{rem:rate_reliability}
$R(\epsilon)$ is an increasing function of $\epsilon$, and hence the transmission time $T_{tx}(\epsilon)$.      
\end{remark}

\subsection{Total Uplink Latency} 

Denote the volume of the data at the device by $D_\text{f}$, which gets compressed at the device itself. 
The volume of compressed data $D_\text{c}$ is given as
\begin{equation}
D_\text{c}(Q) = \frac{D_\text{f}}{Q}. 
\end{equation}

From \eqref{eq:tx_time}, the time elapsed in transmitting the compressed data with outage $\epsilon$ is given as
\begin{equation}
T_\text{tx}( \epsilon) = \frac{D_\text{c}}{R(\epsilon)}T_0  = \frac{D_\text{f}}{Q R(\epsilon)} T_0 .
\end{equation}
The total latency incurred in compression and transmission is 
\begin{equation}
T (Q, \epsilon) =  T_\text{c}(Q) + T_\text{tx}(\epsilon) 
\end{equation}

$T$ is a RV, because $T_\text{c}$ is a RV, see \eqref{eq:t_mec}. 
Using RV transformation techniques, the distribution of $T$ is obtained as  
\begin{align} \label{eq:t_total}  
    f_{T}(t; Q, \epsilon)  =  &  
     \left( \frac{f_\text{R}}{D_f \frac{C(Q)}{\kappa}} \right)^\kappa \frac{1}{\Gamma(\kappa)} \left(t - \frac{D_\text{f}}{Q R(\epsilon)}T_0 \right)^{\kappa-1} \times \nonumber \\
     & \exp \left( -\frac{ \left( t - \frac{D_\text{f}}{Q R(\epsilon)}T_0 \right) f_\text{R}}{D_f  \frac{C(Q)}{\kappa}} \right)   .
\end{align}
The expected latency $\bar{T}$ of the compression and transmission process is given as
\begin{align} \label{eq:lat_exp} 
\bar{T}(Q, \epsilon) =  \mathbb{E}[T(Q, \epsilon)] 
 =  D_\text{f} \left[  \frac{C(Q)}{f_\text{R}} + \frac{T_0}{Q R(\epsilon)} \right ] .
\end{align}
The CDF of $T$ is given as
\begin{equation} \label{eq:cdf_gamma}
F_T(t) =   \frac{\Gamma(\kappa,  \tau(t))}{\Gamma(\kappa)} \;\; \mbox{with}\;\; \tau(t) =  \frac{t - \frac{D_\text{f}}{Q R(\epsilon)}T_0}{\frac{ D_\text{f} C(Q)}{\kappa f_\text{R} }}  
\end{equation}
where $\Gamma(s,x)=\int_{0}^{x} t^{s-1} e^{-t} dt$ is the lower incomplete Gamma function.\footnote{The lower incomplete gamma function is usually denoted by $\gamma (s,x)$. However, use of this notation would introduce ambiguity with the notation used in the paper to denote SNR.}

\begin{theorem}
The CDF $F_T(t)$ is neither a convex nor a concave function of  $t$. 
\end{theorem}
\begin{proof}
See Appendix A.
\end{proof}

\section{Latency-Optimization Framework}
\label{sec:optimization} 
In this section, we consider an optimization framework in which we assess reliability, defined as the probability that the data is received correctly within certain deadline.
Based on the results in Section~\ref{sec:analysis} and on this reliability criterion, two optimization problems are formulated to investigate the trade-off between outage and compression ratio, as discussed below.


\subsection{Outage-Constrained Uplink} 
In some application scenarios, there is a maximum outage level, say $\epsilon_\text{th}$, that can be tolerated for data transmission, with the known SNR statistics as well as device's computational capability at the receiver. 
Here, the total uplink latency elapsed is minimized by maintaining the tolerated outage level. 
The optimization problem in such scenarios is formulated as
\begin{equation}
\nonumber 
\begin{aligned}
 & \textbf{P1}: \min_{N_\text{tx}, Q}  \;  t    \\
 & \text{s.t.} \;  {\textbf{C1}}:  F_{T}(t) \geq  \rho_\text{th}; \; {\textbf{C2}}:  0 \leq  \epsilon  \leq \epsilon_\text{th}; \; {\textbf{C3}}:  Q \geq 1 \\
\end{aligned}  
\end{equation}
Constraint  ${\textbf{C1}}$ expresses the stochastic nature of the latency, where $\rho_\text{th}$ is the probability that the latency should be at most $\tau$.
Constraint  ${\textbf{C2}}$ restricts the outage level not to exceed $\epsilon_\text{th}$, whereas ${\textbf{C3}}$ indicates the range of compression ratio. 

$\textbf{P1}$ is not a convex optimization problem, since $\textbf{C1}$ is not a convex function (see Theorem 1).  
Now, the inverse of ${\textbf{C1}}$, which is the CDF of the Gamma distribution, is given as
\begin{equation}  \label{eq:inv_gamma}
\tau(t) = F_{T}^{-1}(\rho_\text{th}, \kappa) = \tau_0 .
\end{equation}
Using \eqref{eq:cdf_gamma} and \eqref{eq:inv_gamma}, the latency $t$ is given as
\begin{equation}  \label{eq:mod_P1}
t = \frac{\tau_0 D_\text{f}}{\kappa f_\text{R}} C(Q) + D_\text{f} T_0 \frac{1}{Q R(\epsilon)}  =  W(Q, \epsilon). 
\end{equation}

Using \eqref{eq:mod_P1}, the optimization problem $\textbf{P1}$ can be re-interpreted as follows
\begin{equation}
\nonumber 
\begin{aligned}
 & \textbf{P1a}: \min_{Q, \epsilon}  \;  W(Q, \epsilon), \;\; \text{s.t.} \; {\textbf{C2}}  \;   \mbox{and} \;  {\textbf{C3}}  
\end{aligned}  
\end{equation} 

\begin{theorem}
 $W(Q, \epsilon)$  is a convex function of $Q$ and $\epsilon$. 
\end{theorem}
\begin{proof}
See Appendix B.
\end{proof}

Thus, the optimization problem $\textbf{P1a}$ is a convex optimization problem that can be solved by CVX~\cite{grant2014cvx},
and the obtained solution is a global optimum.
The optimal values of $\epsilon$ and  $Q$ can be used to estimate the optimal number of channel use required to transmit the compressed data from \eqref{eq:channel_use}.   

\subsection{Latency-Constrained Uplink}

In another type of scenarios of a fixed latency budget, say $T_\text{th}$, is allotted.
The optimization problem for this case is to minimize the outage, formulated as follows
\begin{equation}
\nonumber 
\begin{aligned}
\textbf{P2}: \min_{N_\text{tx}, Q}  \;  \epsilon, \; \; \text{s.t.} \; {\textbf{C3}}; \; {\textbf{C4}}: F_{T}(t=T_\text{th}) \geq  \rho_\text{th}; 
             {\textbf{C5}}: \epsilon \geq 0
\end{aligned}  
\end{equation} 
Constraint ${\textbf{C3}}$ has been already introduced and indicates the range of compression ratio.
${\textbf{C4}}$ ensures the reliability of transmission in stochastic sense within the latency budget $T_\text{th}$.
${\textbf{C5}}$ restricts the outage level to be greater than $0$. 

Again, $\textbf{P2}$ is not a convex optimization problem, since $\textbf{C4}$ is not a convex function (Theorem 1). 
Using the inverse of the CDF of the Gamma distribution, ${\textbf{C4}}$ can be rewritten as
\begin{equation} \label{eq:t_th_v}
 \frac{1}{R(\epsilon)}   =  \frac{Q}{D_\text{f} T_0} 
\left(  T_\text{th} - \frac{\tau_0 D_\text{f}}{\kappa f_\text{R}} C(Q) \right)  = U(Q)  
\end{equation}
where $\tau_0 = \tau(T_\text{th})$.

Minimizing $\epsilon$ is equivalent to maximizing $\frac{1}{R(\epsilon)}$ since $R(\epsilon)$ is an increasing function of $\epsilon$ (see Remark 1).
Thus, using \eqref{eq:t_th_v}, the optimization problem $\textbf{P2}$ becomes 
\begin{equation}
\nonumber 
\begin{aligned}
 & \textbf{P2a}: \max_{Q}  \;   U(Q), \;\; \text{s.t.} \;   {\textbf{C3}}   \; \mbox{and} \;   {\textbf{C4}}  
\end{aligned}  
\end{equation} 

\begin{theorem}
 $U(Q)$  is a concave function of $Q$. 
\end{theorem}
\begin{proof}
See Appendix C.
\end{proof}

The problem $ \textbf{P2a}$ is concave, and its solution will provide the global optimum. 
The optimal value of outage $\epsilon_\text{opt}$ can be found using the optimal value of $Q$ from \eqref{eq:t_th_v} and  \eqref{eq:rate_epsilon}.

\begin{figure}[!h]
\center  
\subfigure[Latency]
{{\includegraphics[width=0.75\columnwidth]{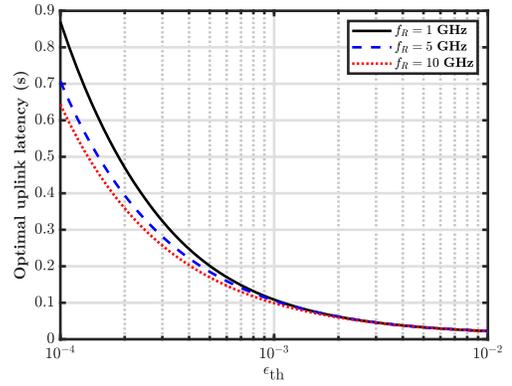}}}
\quad 
\subfigure[Compression ratio]
{{\includegraphics[width=0.75\columnwidth]{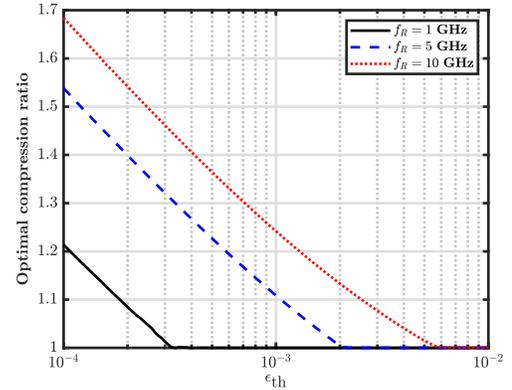}}}
 \quad 
 \subfigure[ Latency fraction for compression ]
{{\includegraphics[width=0.75\columnwidth]{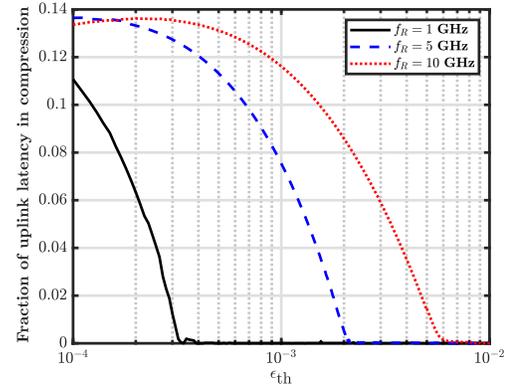}}} 
\caption{ Variation of the optimal system parameters against outage for outage-constrained transmission with $\rho_\text{th}=0.95$. }
\label{fig:reliability_aware_tx} 
\end{figure}

\section{{Numerical Results}}
\label{sec:evaluation}


We illustrate the analysis presented in the previous sections through numerical evaluations.
The values of the parameters are: $\mathcal{K}_0=-27$~dB, $d = 2$~km, $B=10$~MHz, $N_0=-110$~dBm, $T_0=0.5$~$\mu$s, $P_\text{tx}=0.5$~W, $\kappa=1.5$, $\Psi=3.5$, $D_\text{f}=1$~Mb, and $f_\text{R}=\{1, 5, 10 \}$~GHz. 

\subsection{Optimal System Design}

For outage-constrained system design, the optimal system parameters as functions of the outage probability threshold $\epsilon_\text{th}$ are shown in Fig.~\ref{fig:reliability_aware_tx}, for different clock speed of the device’s processor, i.e., $f_\text{R}$.
Fig.~\ref{fig:reliability_aware_tx}(a) reveals that the optimal latency decreases with increase in outage threshold $\epsilon_\text{th}$. 
Likewise, the optimal compression ratio also decreases with the increase in outage threshold, as shown in Fig.~\ref{fig:reliability_aware_tx}(b).
Specifically, a lower tolerated outage demands for a lesser rate and, thus, a higher compression ratio.
When the the tolerated outage is increased, the rate can be increased and, thus, the compression ratio can be decreased.  
Fig.~\ref{fig:reliability_aware_tx}(b) also reveals that transmission without compression (i.e., $Q=1$) becomes the optimal strategy as $\epsilon_{th}$ increases; the value of $\epsilon_\text{th}$ when this happens depends on the $f_\text{R}$.
This is due to the fact that for higher $\epsilon_\text{th}$ it becomes more opportune to spend the time on transmission with higher rate than to spend it on compression.
The same fact is also illustrated in Fig.~\ref{fig:reliability_aware_tx}(c), showing the fraction of uplink latency used in compressing the data as function of outage threshold. 
It can be concluded that compression of raw data is \emph{not} always beneficial, but depends on the capabilities of the processor as well as on the tolerated outage level.


\begin{figure}[!t] 
\center 
\subfigure[Outage]
{{\includegraphics[width=0.75\columnwidth]{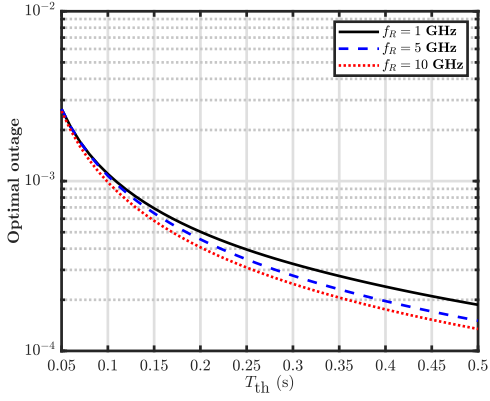}}}
\subfigure[Compression ratio]
{{\includegraphics[width=0.75\columnwidth]{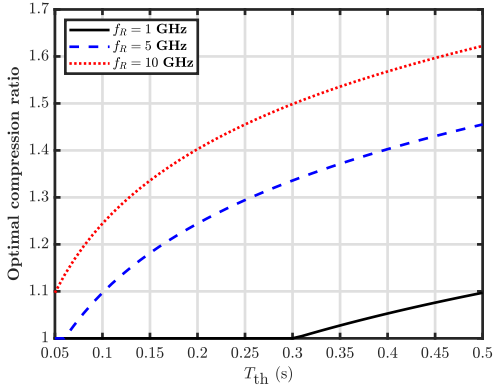}}}
\caption{ Variation of the optimal system parameters against latency budget $T_\text{th}$ for latency-constrained transmission with $\rho_\text{th}=0.95$. } 
\label{fig:latency_aware_tx} 
\end{figure}

For latency-aware system design, the optimal system parameters as functions of latency threshold $T_\text{th}$ are shown in Fig.~\ref{fig:latency_aware_tx}, for varying $f_\text{R}$.
The optimal value of the outage decreases with the increase in $T_\text{th}$, whereas the optimal compression ratio increases.
In other words, when the latency budget is high, more time can be invested in compression, which will lower the rate and allow for a transmission with a lower outage.  


\begin{figure}[!t]
\center 
\subfigure[Outage]
{{\includegraphics[width=0.75\columnwidth]{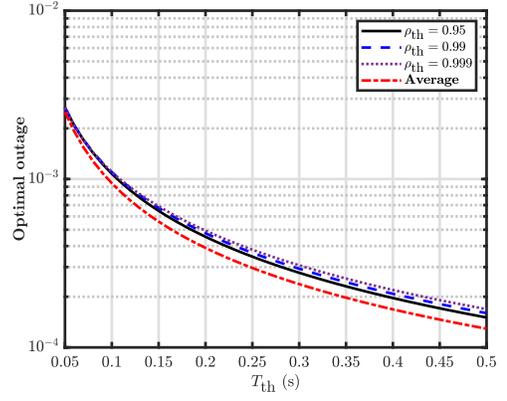}}} 
\quad 
\subfigure[Compression ratio]
{{\includegraphics[width=0.75\columnwidth]{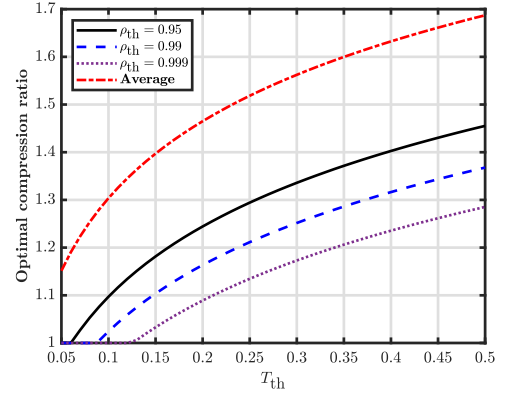}}}
\caption{ Comparison of proposed framework with expected sense for latency-aware transmission  with  $f_R=5$ GHz. } \label{fig:avg_lat_awr} 
\end{figure}

\subsection{Comparison with  Optimization in Expected Sense}
The works reported in ~\cite{onu_1,onu_3,NFV_JSAC,IoT_TI,MEC_TII, hybrid_caching_TI, TelSurg} analyze the latency in teleoperation systems in terms of its average value (i.e., in expected sense), in contrast to the approach taken in this paper. 
For the sake of comparison, here we reformulate  $\textbf{P2}$ to assume expected value of the uplink latency, and examine the obtained results.
For a given latency budget $T_\text{th}$, using \eqref{eq:lat_exp}) we get 
\begin{equation} 
\frac{1}{R(\epsilon)} =  \frac{Q}{T_0}  \left[ \frac{T_{th}}{D_f} - \frac{C(Q)}{f_R}  \right] = V(Q)
\end{equation}  
Thus, the optimization problem for latency-constrained uplink is formulated as follows 
\begin{equation}
\nonumber 
\begin{aligned}
& \textbf{A}: \max_{Q}  \; V(Q),  \;\; \text{s.t.} \;         {\textbf{C3}} \; \text{and} \; {\textbf{C4}} 
\end{aligned}  
\end{equation} 
$V(Q)$ is the concave function of $Q$; the proof is omitted due to space constraint. 
Thus, $\textbf{A}$ is a convex optimization problem that can be solved using CVX.
The optimal system parameters as functions of the latency threshold $T_\text{th}$ are shown in  Fig.~\ref{fig:avg_lat_awr} for (i) different values of $\rho_\text{th}$ for optimization problem $\textbf{P2a}$ and (ii) for problem $\textbf{A}$.
It may be noted that, when $T_\text{th}$ is fixed, the optimal outage in case of optimization in expected sense is lower than that with $\rho_\text{th} \geq 0.95$, see Fig.~\ref{fig:avg_lat_awr}(a).
This implies that in this case, the device will transmit with a lower rate and a higher compression ratio, as shown in Fig.~\ref{fig:avg_lat_awr}(b).
In other words, this approach may lead to over-provisioning.
Conversely, the required latency budget to achieve certain level of outage will be shorter for the system design in expected sense, than for the one treating latency as a RV.
For instance, to achieve an outage of $2 \times 10^{-4}$, the former approach should dimension the latency budget to be $347$ ms, whereas  the latency budget of $392$ ms, $418$ ms, and $432$ ms is required in the latter approach to achieve the reliability of $0.95, 0.99,$ and $0.999$, respectively.
In effect, this represents a case of under-provisioning and of a potential performance degradation.
We also note that that the system design treating outage in the expected sense (i.e., an optimization analogous to the one in $\textbf{P1a}$) will show similar shortcomings; the presentation of the corresponding results is omitted due to space constraints. 

\section{Concluding Remarks}
\label{sec:conclusions} 
{This work has been motivated by the general question about the optimal compression/transmission strategy in systems constrained by latency and reliability.
We have} introduced a framework to analyse the uplink latency of data transfer from the device to the Base Station that has a Mobile Edge Computing server. The data is compressed before  transmission.
We have analyzed the latency as a random variable 
and investigated different trade-offs and achievable performance between latency, link outage, and transmission reliability.
We have also shown the shortcomings of the design approaches that treat latency via its expected value.
Our future work includes the latency analysis of the closed-loop control systems, for which the analysis presented in this paper constitutes a building block.

\section*{Ackowledgment}
 This work is supported by the  European Horizon 2020 project    Tactility  (grant agreement number 856718).

\ifCLASSOPTIONcaptionsoff
  \newpage
\fi

\appendix

\subsection{Proof of Theorem 1}
\label{lemma_1} 

The second derivative of $F_T(t)$ is given as 
{\footnotesize{ 
\begin{align}
\nonumber 
\frac{d^2 F_T(t)}{dt^2} = & \frac{1}{\Gamma(\kappa)} \left(e^{-(t - \frac{D_\text{f}}{Q R(\epsilon)}T_0)}\right)  \left(t - \frac{D_\text{f}}{Q R(\epsilon)}T_0\right)^{\kappa-1} \times \\
& \left( -1 + \frac{\kappa-1}{t - \frac{D_\text{f}}{Q R(\epsilon)}T_0} \right) \nonumber
\end{align}
}} 
Here $t - \frac{D_\text{f}}{Q R(\epsilon)}T_0 > 0$  because the domain of definition of the Gamma distribution is positive. 
Thus, $\frac{d^2 F_T(t)}{dt^2} > 0$ for $t < \kappa -1 + \frac{D_\text{f}}{Q R(\epsilon)} T_0 $ and $\frac{d^2 F_T(t)}{dt^2} < 0$ for $t > \kappa -1 + \frac{D_\text{f}}{Q R(\epsilon)}T_0$, and $F_T(t)$ is neither a convex nor a convex function of $t$.

\subsection{Proof of Theorem 2}
\label{lemma_1} 

The Hessian matrix of $W(Q, \epsilon)$ is given as 
{\footnotesize{  
\begin{equation}
\nonumber 
 \mathcal{H}_1  = 
\begin{bmatrix}
\frac{\delta^2 W }{\delta Q^2}  & \frac{\delta^2 W }{\delta Q \delta \epsilon}  \\
 \frac{\delta^2 W }{\delta \epsilon \delta Q}  &  \frac{\delta^2 W }{\delta \epsilon^2}
\end{bmatrix} 
\end{equation} 
}} 
As mentioned in Remark \ref{rem:rate_reliability}, $R(\epsilon)$ is an increasing function of $\epsilon$. Therefore, for the purpose of this analysis, differentiating with respect to $\epsilon$ is the same as differentiating with respect to $R(\epsilon)$.
Thus, we can write the Hessian matrix as follows
{\footnotesize{
\begin{equation}
\nonumber 
 \mathcal{H}_2  = 
\begin{bmatrix}
\frac{\delta^2 W }{\delta Q^2 }   & \frac{\delta^2 W  }{\delta Q \delta R(\epsilon)} \\
 \frac{\delta^2 W }{\delta R(\epsilon)  \delta Q}  &  \frac{\delta^2 W }{\delta R(\epsilon)^2} 
\end{bmatrix} 
\end{equation} 
}} 
The elements of Hessian matrix are given as 
{\footnotesize{
\begin{equation} 
\nonumber 
\begin{split}
\frac{\delta^2 W }{\delta Q^2} & = D_f \left( \frac{\tau_0}{\kappa f_\text{R}} \Psi^2  \exp(\Psi Q) + \frac{2 T_0}{Q^3 R(\epsilon)}  \right) \\
\frac{\delta^2 W  }{\delta Q \delta \epsilon} & =  \frac{\delta^2 }{\delta \epsilon \delta Q} t   = \frac{D_\text{f} T_0}{Q^2 R^2(\epsilon)} \\
\frac{\delta^2 W }{\delta R(\epsilon)^2}  & =  \frac{2 D_\text{f} T_0}{Q R^3(\epsilon)} \\ 
\end{split}
\end{equation} }}
The determinant of Hessian matrix is given as {\footnotesize{
\begin{equation} 
\nonumber 
\begin{split}
|\mathcal{H}_2|  & = \frac{\delta^2 W }{\delta Q^2}  \frac{\delta^2 W }{\delta R(\epsilon)^2} - \left( \frac{\delta^2 W }{\delta Q \delta \epsilon} \right)^2   \\
 & =   \frac{D_\text{f}^2 T_0}{Q R^3(\epsilon)} \left(  \frac{2 \tau_0}{\kappa f_\text{R}} \Psi^2 \exp(\Psi Q) + \frac{3 T_0}{Q^3 R(\epsilon)} \right)    
 \end{split}
\end{equation} }}
Observe that $|\mathcal{H}_2|$ is always positive and hence $W(Q, \epsilon)$ is a convex function of $Q$ and $\epsilon$.

\subsection{Proof of Theorem 3}
\label{lemma_1}

The first derivative of $U(Q)$ is given as {\footnotesize{
\begin{equation}
\nonumber 
\frac{d U(Q)}{dQ}  = \frac{1}{D_\text{f} T_0} \left[   T_\text{th} - \frac{\tau_0 D_\text{f}}{\kappa f_\text{R}} \left( (1+Q \Psi) \exp(Q\Psi) + \exp(\Psi)  \right)  \right]
 \end{equation} }}
The second derivative of $U(Q)$ is given as {\footnotesize{
\begin{equation}
\nonumber 
\frac{d^2 U(Q)}{dQ^2}  =  -\frac{\tau_0}{\kappa T_0 f_\text{R}} (2 \Psi + Q \Psi^2) \exp(\Psi Q) 
 \end{equation} }}
Observe that $\frac{d^2 U(Q)}{dQ^2} <0$, which proves that $U(Q)$ is a concave function of $Q$.

\bibliographystyle{IEEEtran}
\bibliography{ref_final}

\end{document}